\documentclass[journal]{IEEEtran}

\usepackage{amsmath,graphicx}
\usepackage{amsfonts} \usepackage{color}
\usepackage{amssymb}
\usepackage{amsthm}
\usepackage{cite}
\usepackage{array}
\newtheorem{lemma}{Lemma}

\usepackage{algorithm,algorithmic}
\usepackage{dsfont}
\makeatletter
\newcommand*{\rom}[1]{\expandafter\@slowromancap\romannumeral #1@}
\makeatother
\ifCLASSINFOpdf
\else
\fi
\begin{document}
%
\title{Partial Consensus and Conservative Fusion of Gaussian Mixtures for Distributed PHD Fusion}
%
%
%

\author{Tiancheng~Li,
        ~Juan~M. Corchado, 
        and~Shudong Sun 
\thanks{Manuscript received xxxx}
\thanks{T. Li and J.M. Corchado are with School of Sciences, University of Salamanca, 37007 Salamanca, Spain, E-mail: \{t.c.li, corchado\}@usal.es; During the work, T. Li has undertaken a Secondment at the Institute of Telecommunications, Vienna University of Technology, 1040 Wien, Austria}
\thanks{S. Sun is with School of Mechanical Engineering, Northwestern Polytechnical University, Xi’an 710072, China, E-mail: sdsun@nwpu.edu.cn}
\thanks{This work is in part supported by the Marie Sk\l{}odowska-Curie Individual Fellowship (H2020-MSCA-IF-2015) with Grant no. 709267}}
%
%

\markboth{Nov. 20,~2017}%
{Li \MakeLowercase{\textit{et al.}}: Partial Consensus and Conservative Fusion of Gaussian Mixtures for Distributed PHD Fusion}
%



\maketitle

\begin{abstract}
We propose a novel consensus notion, called ``partial consensus'', for distributed GM-PHD (Gaussian mixture probability hypothesis density) fusion based on a peer-to-peer (P2P) sensor network, in which only highly-weighted posterior Gaussian components (GCs) are disseminated in the P2P communication for fusion while the insignificant GCs are not involved. The partial consensus does not only enjoy high efficiency in both network communication and local fusion computation, but also significantly reduces the affect of potential false data (clutter) to the filter, leading to increased signal-to-noise ratio at local sensors. Two ``conservative'' mixture reduction schemes are advocated for fusing the shared GCs in a fully distributed manner. One is given by pairwise averaging GCs between sensors based on Hungarian assignment and the other is merging close GCs based a new GM merging scheme. The proposed approaches have a close connection to the conservative fusion approaches known as covariance union and arithmetic mean density. In parallel, average consensus is sought on the cardinality distribution (namely the GM weight sum) among sensors. Simulations for tracking either a single target or multiple targets that simultaneously appear are presented based on a sensor network where each sensor operates a GM-PHD filter, in order to compare our approaches with the benchmark generalized covariance intersection approach. 
The results demonstrate that the partial, arithmetic average, consensus outperforms the complete, geometric average, consensus. 
\end{abstract}


%
\IEEEpeerreviewmaketitle

\section{Introduction}
%
%
%
%
\IEEEPARstart{T}{he} rapid development of wireless sensor networks (WSNs) in the last decade is in large part responsible for the recent upsurge in interest in WSN-based distributed tracking.
A typical decentralized WSN consists of a number of spatially distributed, interconnected sensors that have independent sensing and calculation capabilities and (only) communicate with the neighbors for information sharing, namely peer-to-peer (P2P) communication. 
In particular, due to the appealing fault-tolerance and scalability to large networks, consensus-based distributed algorithms have gained immense popularity in the sensor networks community.

In the consensus-oriented distributed filtering setup, each sensor operates an independent filter while sharing information with its neighbors iteratively to ameliorate each other's estimation with the goal
of converging to the same estimate over the entire network. As a result, local estimation that are made based on both local observation and the information disseminated from neighbors are similar to each other (or in other words, the sensors asymptotically reach a consensus), which are better as compared to the independent estimation without network cooperation \cite{Liu07,Akselrod12,Mohammadi15}. In this paper, we consider the scenario with a time-varying, unknown, number of targets, which are synchronously observed by all sensors in the presence of false and missing observations. 

Great interest has been seen for extending the theory of average consensus \cite{Xiao04,Olfati-Saber07} for which the item being estimated may be the arithmetic average (considered as the default manner \cite{Olfati-Saber07}, akin to the linear opinion pool \cite{Heskes98}
) or the geometric average \cite{Olfati-Saber06} (akin to the logarithmic opinion pool \cite{Heskes98}
) of the initial values. 


With regard to the type of information disseminated, three main categories of protocols exist; we note that there are protocols such as diffusion \cite{sayed14,Dedecius17} that may belong to either. Our approach falls into the last category:

\begin{enumerate}
\item \textit{Measurement/Likelihood}. Disseminating raw measurement can be practically preventable in communication. Instead, the likelihood function, as a compact representation of the measurement information, is a promising alternative \cite{Hlinka12, Hlinka13, Battistelli14b}. However, in multi-sensor multi-target cases, computationally cumbersome measurements-to-targets association or enumeration \cite{Mahler10} is typically required. Moreover, it is nontrivial to fuse raw measurements reported at sensors of different profiles including detection probabilities, clutter rates, etc. To date, measurement/likelihood consensus is mainly limited to the single target case. 
\item \textit{Estimate/Track}. This involves running tracking algorithms on each sensor separately, yielding a set of tracks to be associated between sensors based on their proximities and then be fused, namely track-to-track fusion \cite{Mori12, Noack14}. When tracks are distant in the state space, this may work well, e.g., \cite{Beaudeau15, Zhu17} otherwise it suffers from the fragility and high computational cost for maintaining a large number of hypotheses. 
\item \textit{Posterior/Intensity}. This involves disseminating and fusing the multi-target posterior \cite{Meyer14, Meyer17} or the density of its statistical moments between sensors. 
In particular, the probability hypothesis density (PHD) that is the first order moment of the random target-state set, has been developed as a powerful alternative to the full posterior for time series recursion \cite{Mahler03, Vo06}. In this case, the key is to disseminate and fuse PHDs.
\end{enumerate}

As the state of the art, the geometric average for PHDs/multi-target densities is referred to as the Kullback-Leibler average (KLA) \cite{Battistelli13, Battistelli14a, Battistelli14b}, also known as the geometric mean density (GMD) or the exponential mixture density (EMD) \cite{Uney11, Uney13, Bailey12}. The fusion approach is known as generalized covariance intersection (GCI) \cite{Mahler00, Clark10, Mahler12} which extends the Chernoff fusion/covariance intersection \cite{Uhlmann95,Julier97} to multi-target densities. In spite of its success in certain scenarios, deficiencies of GCI, most of which have already been recognized in the literature, are noted in this paper. 


However, it is not our intention to revise or improve these geometric average consensus approaches. Instead, we propose novel arithmetic average consensus approaches for PHD fusion, which are closely connected to the so-called covariance union (CU) \cite{Uhlmann03,Julier05,Bochardt06,Reece10} and arithmetic mean density (AMD) \cite{Bailey12}. 
In short, there are two distinguishable features with our approaches:
\begin{enumerate}
\item Only the significant components of local PHDs, which are more likely target signals rather than false alarms, are disseminated between neighbors, while the insignificant components that are more suspected to be false alarms will not be involved in either the P2P communication or the consensus fusion. As such, the signal-to-noise ratio (SNR) can be positively enhanced at local sensors. This significantly differs from existing consensus approaches where the (complete) consensus is conditioned on \textit{all} the information available in the network. The consensus in our approach is made based on a part of the information of local posteriors, termed \textit{partial consensus}. %
\item Only closely distributed components, which are more likely corresponding to the same target, are fused, in either of two conservative manners: averaging and merging, based on union calculation rather than intersection. The resulting  consensus remains defined in the default \textit{arithmetic average} sense rather in the KLA sense, which demonstrates particular advantages in dealing with the false and missing observations which are inevitably involved in realistic tracking.
\end{enumerate}

A preliminary part of the merging-for-fusion protocol has been presented in our conference paper \cite{Li17merging_Conf}, in which, however, we did not analyze its connection to AMD/CU, nor did we provide any conservativeness analysis. The merging scheme adopted initially is a standard one \cite{Salmond90}, which as found in this paper can be optimized in the covariance-fusion part. 
These have now been completed in this article. In addition, we present much more technical extension and new results, including a new, communicatively much cheaper, partial consensus protocol. Therefore, this paper serves as a significant revision and extension to \cite{Li17merging_Conf}. 

The remainder of this paper is organized as follows. The basics of GM-PHD, conservative fusion and GCI are given in Section \ref{sec:background}, followed by the motivation and key idea of our approach in section \ref{sec:Motivation}. The proposed distributed GM-PHD fusion protocol is detailed in Section \ref{sec:protocol}. Simulations are given in Section \ref{sec:simulation} for comparing our approaches with the GCI. In particular, the weakness of the GCI is noted
. We conclude in Section \ref{sec:conclusion}.

\section{Background and Preliminary}
\label{sec:background}
\subsection{RFS and GM-PHD}

We consider an unknown and time-varying number $M_k$ of targets with random states $\mathbf{x}_k^{(n)}$ in the state space $\chi \subseteq \mathbb{R}^d$, $n = 1,2,\cdots,M_k$. The collection of target states at time $k$ can be modeled by a \emph{random finite set} (RFS) $X_k=\left\{{\mathbf{x}_{k,1},\cdots,\mathbf{x}_{k,M_k}}\right\}$ with random cardinality $M_k =|X_k|$. The cardinality distribution $\rho(n)$ of $X_k$ is a discrete probability mass function of $M_k$, i.e., $\rho(n) \triangleq \mathrm{Pr}[M_k \!=\! n]$.

A RFS variable $X$ is a random variable that takes values as unordered finite sets and is uniquely specified by its cardinality distribution $\rho(n)$ and a family of symmetric joint distributions $p_n (\mathbf{x}_1,\mathbf{x}_2,\cdots,\mathbf{x}_n)$ that characterize the distribution of its elements over the state space, conditioned on the set cardinality $n$. Here, a joint distribution function $p_n (\mathbf{x}_1,\mathbf{x}_2,\cdots,\mathbf{x}_n )$  is said to be symmetric if its value remains unchanged for all of the $n!$ possible permutations of its variables. The probability density function (PDF) $f(X)$ of a RFS variable $X$ is given as $f(\{\mathbf{x}_1,\mathbf{x}_2,\cdots, \mathbf{x}_n \})=n! \rho(n) p_n (\mathbf{x}_1,\mathbf{x}_2,\cdots, \mathbf{x}_n )$. 


Instead of propagating the full multi-target density which has been considered computationally intractable, the PHD filter propagates its first order statistical moment \cite{Mahler03}. For a multi-target RFS variable $X$ with the PDF $f(X)$, its first order moment PHD $D_S(\mathbf{x})$ in a region $S \subseteq \chi$ is given as:
\begin{equation}\label{eq:PHD}
D_S(\mathbf{x})=\int_S \delta_X(\mathbf{x})f(X)\delta X \hspace{0.5mm},
\end{equation}
where $\delta_X(\mathbf{x}) \triangleq \sum_{w\in X}\delta_w(\mathbf{x})$ which is used to convert the finite set $X=\{\mathbf{x}_1,\mathbf{x}_2,\cdots \}$ into vectors since the first order moment is defined in the single-target vector space, 
$\delta_\mathbf{y}(\mathbf{x})$ denotes the generalized Kronecker delta function, 
and the RFS integral in the region $S \subseteq \chi$ is defined as:
\begin{align}\label{eq:SetInt}
& \int_S f(X)\delta X\nonumber \\
&  \triangleq f(\emptyset) + \sum_{n=1}^\infty \int_{S^n} \frac{f(\{\mathbf{x}_1,\mathbf{x}_2,\cdots, \mathbf{x}_n\})}{n!} d\mathbf{x}_1d\mathbf{x}_2\cdots d\mathbf{x}_n \hspace{0.5mm}. 
\end{align}



Straightforwardly, the GM approximation of the whole PHD 
at filtering time $k$ can be written as:
\begin{equation}\label{eq:GMPHD}
D_{k}(\mathbf{x})= \sum_{i=1}^{J_k} w_k^{(i)} \mathcal{N}(\mathbf{x};\mathbf{m}_k^{(i)},\mathbf{P}_k^{(i)}) \hspace{0.5mm},
\end{equation}
where $\mathcal{N}(\mathbf{x;m,P})$ denotes a Gaussian component (GC) with mean $\mathbf{m}$ and covariance $\mathbf{P}$, $J_k$ is the number of GCs in total, and $w_k^{(i)}$ is the weight of $i$th GC. 

The PHD is uniquely defined by the property that its integral in any region gives the expected number of targets in that region. Therefore, the expected number of targets 
can be approximated by the weight sum $W_k$  of all GCs, i.e., 
\begin{equation}\label{eq:4}
W_k = \sum_{i=1}^{J_k} w_k^{(i)} \hspace{0.5mm}.
\end{equation}


In this paper, we assume each local sensor running a GM-PHD filter, e.g., as given in \cite{Vo06} and that they are synchronous. Our work focuses on the posterior GM dissemination and fusion between neighboring sensors. 

\subsection{Conservative Fusion and Mixture Reduction}

We consider now a sensor network where all the sensors observe the same set of targets but their measurements are conditionally independent. The errors of estimates yielded at local sensors are correlated with each other where the correlation is due to the common prior/noise in the models and common information shared between sensors, etc. 
Despite any a priori information on the cross-covariance \cite{Chang97,Chong04, Gao16}, it is typically intractable, if not impossible, to quantify the correlation between sensors, which is at least time-varying due to the P2P communication and prevents optimal fusion (e.g., in the sense of minimum mean square error); see also \cite{Bakr17}. Therefore, a pseudo-optimal, ``conservative", fusion rule is resorted to for 
avoiding underestimating the actual squared estimate errors. The benefit to do so includes better fault tolerance and robustness \cite{Uhlmann03,Julier05}. To be more precise, we have the following definition on the notion of ``conservative'', as used in \cite{Uhlmann95,Uhlmann03,Bochardt06,Ajgl14}:

\textit{Definition 1 (conservativeness)}. An estimate pair ($\mathbf{\hat{x}},\mathbf{P}$) of the real state $\mathbf{x}$ (a random vector), consisting of a vector estimate mean $\mathbf{\hat{x}}$ with an associated error covariance matrix $\mathbf{P}$, is called conservative when $\mathbf{P}$ is no less than the actual error covariance of the estimate, i.e., $\mathbf{P} - \mathrm{E} [(\mathbf{x}-\mathbf{\hat{x}})(\mathbf{x}-\mathbf{\hat{x}})^\mathrm{T}]$ is positive semi-definite. 

With the associated covariance matrix being ``conservative'', the estimate pair is also called ``consistent'' \cite{Uhlmann03,Reece10}. However, a consistent estimator is traditionally an estimator that converges in probability to the quantity being estimated as the sample size grows. To avoid confusion, we shall only use the terminologies of ``conservative'' or ``conservativeness''. 

\begin{lemma} \label{lemma1}
A sufficient condition for the fused estimate pair ($\mathbf{\hat{x}},\mathbf{P}$), due to fusing a set of estimate pairs $(\mathbf{\hat{x}}_i,\mathbf{P}_i), i \in \mathcal{I} =\{1,2,\cdots\}$, in which at least one pair is unbiased and conservative, to be conservative is that
\begin{equation} \label{eq:CU_Sufficient}
\mathbf{P} \geq \mathbf{P}_i + (\mathbf{\hat{x}}-\mathbf{\hat{x}}_i)(\mathbf{\hat{x}}-\mathbf{\hat{x}}_i)^\mathrm{T}, \hspace{1mm} \forall i \in \mathcal{I} \hspace{0.5mm}.
\end{equation}
\end{lemma}
\begin{proof}
Without lose of generality, supposing $(\mathbf{\hat{x}}_i,\mathbf{P}_i)$ is unbiased and conservative, we have 
\begin{equation}
\mathrm{E} [ (\mathbf{x}-\mathbf{\hat{x}}_i)(\mathbf{\hat{x}}_i-\mathbf{\hat{x}})^\mathrm{T} ] =   \mathbf{0} \hspace{0.5mm}, 
\end{equation}
\begin{equation}
\mathbf{P}_i \geq \mathrm{E} [(\mathbf{x}-\mathbf{\hat{x}}_i)(\mathbf{x}-\mathbf{\hat{x}}_i)^\mathrm{T}] \hspace{0.5mm}.
\end{equation}
due to the unbiasedness and conservativeness, respectively.

By decomposing $\mathbf{x}-\mathbf{\hat{x}}_i $ as $(\mathbf{x}-\mathbf{\hat{x}}_i)+(\mathbf{\hat{x}}_i-\mathbf{\hat{x}})$, we obtain $\mathrm{E} [(\mathbf{x}-\mathbf{\hat{x}})(\mathbf{x}-\mathbf{\hat{x}})^\mathrm{T}] \leq \mathbf{P}_i + (\mathbf{\hat{x}}-\mathbf{\hat{x}}_i)(\mathbf{\hat{x}}-\mathbf{\hat{x}}_i)^\mathrm{T}$ easily to finish the proof.
\end{proof}

\begin{lemma} \label{lemma2}
Given a set of conservative estimate pairs $(\mathbf{\hat{x}},\tilde{\mathbf{P}}_i),i \in \mathcal{I}=\{1,2,\cdots\}$ of the same, unbiased estimate mean associated with possibly different error-covariance matrix, a sufficient condition for the fused estimate pair ($\mathbf{\hat{x}},\mathbf{P}$), to be conservative is given by
\begin{equation} \label{eq:lemma2_P}
\mathbf{P} \geq  \sum_{i \in \mathcal{I}} \omega_i \tilde{\mathbf{P}}_i \hspace{0.5mm},
\end{equation}
where the non-negative scaling parameters $\omega_i \geq 0$, $\sum_{i \in \mathcal{I}} \omega_i = 1$ are called fusing weights hereafter. 
\end{lemma}
\begin{proof}
The conservativeness of fusing estimate pairs reads
\begin{equation} \label{eq:Covlemma2}
\tilde{\mathbf{P}}_i \geq \mathrm{E} [(\mathbf{x}-\mathbf{\hat{x}}) (\mathbf{x}-\mathbf{\hat{x}})^\mathrm{T}], \hspace{1mm} \forall i \in \mathcal{I} \hspace{0.5mm}.
\end{equation}
The proof is simply done by multiplying both sides of \eqref{eq:Covlemma2} by $\omega_i$ and summing up over all $i \in \mathcal{I}$, which leads to 
\begin{equation}
\sum_{i \in \mathcal{I}} \omega_i \tilde{\mathbf{P}}_i \geq \mathrm{E} [(\mathbf{x}-\mathbf{\hat{x}}) (\mathbf{x}-\mathbf{\hat{x}})^\mathrm{T}] \hspace{0.5mm}.
\end{equation}
\end{proof}

\textit{Definition 2 (Standard Mixture Reduction, SMR)}. Given a set of estimate pairs $(\mathbf{\hat{x}}_i,\mathbf{P}_i)$ weighted as $w_i, i \in \mathcal{I}$, respectively, the SMR scheme \cite{Salmond90} fuses them into a single estimate pair $(\mathbf{\hat{x}}_\mathrm{SMR}, \mathbf{P}_\mathrm{SMR})$ with weight $w_\mathrm{SMR}$, given by
\begin{equation} \label{eq:GMmerging_w}
w_\mathrm{SMR}=\sum_{i \in \mathcal{I}} w_i \hspace{0.5mm},
\end{equation}
\begin{equation} \label{eq:GMmerging_x}
\mathbf{\hat{x}}_\mathrm{SMR} = \frac{\sum_{i \in \mathcal{I}} w_i\mathbf{\hat{x}}_i}{\sum_{i \in \mathcal{I}} w_i} \hspace{0.5mm},
\end{equation}
\begin{equation} \label{eq:GMmerging_P}
\mathbf{P}_\mathrm{SMR} = \frac{\sum_{i \in \mathcal{I}} w_i\tilde{\mathbf{P}}_i}{\sum_{i \in \mathcal{I}} w_i} \hspace{0.5mm},
\end{equation}
where the adjusted covariance matrix is given by (cf. \eqref{eq:CU_Sufficient})
\begin{equation} \label{eq:CovAdjut}
\tilde{\mathbf{P}}_i=\mathbf{P}_i + (\mathbf{\hat{x}}_\mathrm{SMR}-\mathbf{\hat{x}}_i)(\mathbf{\hat{x}}_\mathrm{SMR}-\mathbf{\hat{x}}_i)^\mathrm{T} \hspace{0.5mm}.
\end{equation}

\begin{lemma} \label{lemma3}
Given that all the fusing estimate pairs are unbiased and conservative, the resulting estimate pair given by the SMR scheme as in \eqref{eq:GMmerging_x}-\eqref{eq:GMmerging_P} is unbiased and conservative.
\end{lemma}

\begin{proof}
The proof is straightforwardly based on Lemmas \ref{lemma1} and \ref{lemma2}. First, unbiasedness is due to the convex combination. Second, given that each $(\mathbf{\hat{x}}_i,\mathbf{P}_i), \forall i \in \mathcal{I} = \{1,2,\cdots \}$ is unbiased and conservative, $(\mathbf{\hat{x}}_\mathrm{SMR},\mathbf{\tilde{P}}_i)$ is conservative according to Lemma \ref{lemma1}, and so their convex combination of $(\mathbf{\hat{x}}_\mathrm{SMR},\mathbf{\tilde{P}}_\mathrm{SMR})$ is conservative according to Lemma \ref{lemma2}.
\end{proof}

It is important to note that, considering ``conservativeness" only, the fusing weights used to get a conservative fused covariance matrix as in \eqref{eq:GMmerging_P} do not have to be the same as that to get the fused state as in \eqref{eq:GMmerging_x}. But instead, it is typically of interest to use different fusing weights to get an optimal fused covariance in some sense, while retaining conservativeness. For this, we have the following proposition, akin to the CI-based optimization \cite{Chen02,Niehsen02}. 

\textbf{Proposition 1}. The covariance-fusing weights for \eqref{eq:lemma2_P} can be determined such that the trace of the resulting covariance matrix is minimized, i.e., $\big\{\omega_i\big\}_{i \in \mathcal{I}} = \underset{\omega_i,i \in \mathcal{I}}{\mathrm{argmin}} \hspace{0.5mm} \mathrm{tr}\Big(\sum_{i \in \mathcal{I}} \omega_i\tilde{\mathbf{P}}_i\Big)$. 
Thanks to the convex combination and positive trace of the matrices, 
the solution is simply given by $\omega_i=1, \omega_j=0, \forall j \neq i, j \in \mathcal{I}$ where $i=\underset{i \in \mathcal{I}}{\mathrm{argmin}} \hspace{0.5mm} \mathrm{tr} (\tilde{\mathbf{P}}_i)$. That is, the trace-minimal yet conservative fused covariance is given by 
\begin{equation} \label{eq:minGM_P}
\mathbf{P}_\mathrm{OMR} = 
\underset{\tilde{\mathbf{P}}_i}{\mathrm{argmin}} \hspace{0.5mm} \mathrm{tr}\big(\tilde{\mathbf{P}}_i\big) \hspace{0.5mm}.
\end{equation}

Hereafter, we refer to the MR scheme based on \eqref{eq:GMmerging_w}, \eqref{eq:GMmerging_x} \eqref{eq:minGM_P} and \eqref{eq:CovAdjut} as the optimal mixture reduction (OMR), which differs from the SMR only in the covariance-fusion part. It is a type of fusion seeking conservativeness, given that all fusing estimate pairs are unbiased and conservative. 

\subsection{GCI Fusion and KLA}
Given a set of posteriors $f_i \in \Psi, i \in \mathcal{I}$ to be fused by the fusing weights $\omega_i \geq 0$, where $\Psi$ denotes the set of PDFs over the state space $\chi$, 
and 
$\mathcal{I}=\{1,2,\cdots\}$ denotes the fusing sensor set, 
the GCI/Chernoff fusion \cite{Mahler00} which resembles the logarithmic opinion pool \cite{Heskes98} and the belief consensus \cite{Olfati-Saber06} reads
\begin{equation}\label{eq:GCI}
f_\mathrm{GCI} \triangleq C^{-1}\prod_{i\in \mathcal{I}}{f_i^{\omega_i}} \hspace{0.5mm},
\end{equation}
where $C$ is a normalization constant.


The GCI fusing result is also known as GMD \cite{Bailey12} or EMD \cite{Julier06, Uney11, Uney13}, which actually minimizes the weighted sum of its KLD with respect to all posteriors $f_i, \forall i\in \mathcal{I}$, and is, therefore, also referred to as the weighted Kullback-Leibler average (KLA) \cite{Battistelli13, Battistelli14b}, i.e., 

\begin{equation}\label{eq:kla}
f_\mathrm{KLA} = \underset{f \in \Psi}{\mathrm{arginf}} \sum_i {\omega_i d_{KL}(f||f_i)} \hspace{0.5mm},
\end{equation}
where $d_{KL} (f_a ||f_b) \triangleq \int f_a (X) \frac{f_a (X)}{f_b (X)}\delta X $ is the set-theoretical KLD of the intensities $f_a$ from $f_b$.

Three challenges arise due to the exponentiation and product calculation when the posterior $f_i$ in \eqref{eq:GCI} is given by a mixture, such as the GM:
\begin{enumerate}
\item The fractional order exponential power of a GM 
does not provide a GM. Existing solutions are based on either analytical approximation that only appeals to special mixtures (e.g., components are well distant) \cite{Battistelli13, Battistelli14a, Battistelli14b, Battistelli15} or numerical approximation via important sampling \cite{Mariam07,Li17Nehorai} or sigma point method \cite{Gunay16}. 
\item The product rule is prone to mis-detection. 
Misdetection at one sensor 
will remarkably degrade the detection at the other sensors since any signal multiplied by a weak signal of almost zero energy will be greatly weakened. See also illustration given in \cite{Yu16,Yi17}. 
\item The GCI/KLA fusion will typically result in a multiplying number of fused GCs \cite{Battistelli13}, which is costly in both communication and computation. 
\end{enumerate}

These problems can lead to disappointing results in certain cases
, which will be discussed within our simulation study in Section \ref{sec:simulation}.C. To overcome these deficiencies, we propose alternatives without the intrinsic need for exponentiation and product calculation of mixtures while being ``conservative'' not only in fusion but also in communication. In the following distributed formulation, we will use subscripts $a$ and $b$ to distinguish between two neighboring sensors. Since all calculations regard the same filtering time $k$, we drop the subscript $k$ for notation simplicity.

\section{Key Idea and Properties of Our Proposal} \label{sec:Motivation}
The section presents two ``conservative'' principles for distributed fusion algorithm design, which constitute the key idea of our approaches:

\begin{itemize}
\item \textit{P.1 Conservative communication}. Consensus should only be sought on the information of targets. To get this maximally respected, only the GC of significance (those that are highly likely to corresponding to the ``target'') should be disseminated, while the insignificant GC (those that are more like false alarms) should be the least involved for conservative consideration. We refer to this as the ``conservative communication'' principle. Here the ``conservativeness'' is not the same as the estimate ``conservativeness'' given in Definition 1. We assume that the reader will not be confused. 
\item \textit{P.2 Conservative fusion}. Only highly relevant information, namely that which corresponds to the same target as at least very likely, should be fused and the fused results should retain ``conservativeness'', to deal with the unknown correlation between sensors. We refer to this as the ``conservative fusion'' principle.
\end{itemize}


%
\subsection{Conservative Communication for Partial Consensus}
First, the mixture reduction is carried out in local GM filters as usual at each filtering step, before network communication, to control the GM size. 


Second, only highly-weighted GCs that are more likely corresponding to real targets rather than false alarms are disseminated between neighbors. To this end, we propose two alternative rules to identify these target-likely GCs, referred to as \textit{rank rule} and \textit{threshold rule}, respectively. 

\begin{itemize}
\item \textit{P.1.1 Rank rule}. Specify the number of GCs to be disseminated as equal to the intermediately estimated number of targets at each sensor using the closest integer to $W_k$ as in \eqref{eq:4} or more straightforwardly, specify a fixed number of GCs when a priori information (e.g., maximum) about the number of targets is known. Then, only the corresponding number of GCs with the greatest weights are transmitted to the neighbors. 

\item \textit{P.1.2 Threshold rule}. Specify a weight threshold $w_s$, and only the component that is weighted greater than that threshold will be transmitted.
\end{itemize}

It is also possible to use a hybrid, arguably more conservative, criterion such that only the GCs that fulfill both rules are chosen, or a hybrid, less conservative, criterion such that any GCs that fulfill either rules can be chosen. In any case, we refer to them as a separate GM, hereafter called Target-likely GM (T-GM) and denote the T-GM size at sensor $a$ as $n_a$, i.e., 
\begin{equation}\label{eq:PHD4}
D_{a,\mathrm{T}}(\mathbf{x})\triangleq \sum_{i=1}^{n_a} w_a^{(i)} \mathcal{N}(\mathbf{x};\mathbf{m}_a^{(i)},\mathbf{P}_a^{(i)}) \hspace{0.5mm},
\end{equation}
of which the total weight ($\leq W_a $) is given as
\begin{equation}\label{eq:PHD5}
W_{a,\mathrm{T}} = \sum_{i=1}^{n_a} w_a^{(i)} \hspace{0.5mm}.
\end{equation}

Correspondingly, the remaining GC at each sensor is called false alarm-suspicious GC (FA-GC), which will not be involved in the neighborhood communication. 

\textit{Definition 3 (partial consensus)} The consensus yielded by disseminating among sensors an incomplete part of the information they own, i.e., only target-likely GCs in our approaches, is called \textit{partial consensus}. 

\subsection{Conservative AMD Fusion}
Different to the KLA optimality of the GMD as in \eqref{eq:kla}, the AMD \cite{Bailey12} calculates the average of posterior multi-target density in the arithmetic sense \cite{Olfati-Saber07} rather than in the geometric average sense \cite{Olfati-Saber06}, or equivalently speaking, based on the linear opinion pool rather than the logarithmic opinion pool.

\textit{Definition 4 (AMD)}. The AMD of multiple posteriors $f_i, i \in \mathcal{I}$, is given as follows:
\begin{equation}\label{eq:AMD}
f_\mathrm{AMD} 
\triangleq \sum_{i \in \mathcal{I}}{\omega_if_i} \hspace{0.5mm},
\end{equation}
where the fusing weights $\omega_i \geq 0, \sum_{i\in \mathcal{I}} {\omega_i} =1$. As addressed, $f_i$ is only a partial PHD obtained at sensor $i$ in our work. 




As shown, the AMD is given by a convex union of multisensor posteriors, which does not double count information \cite{Bailey12} and is provably conservative (cf. Lemma 2)
. It was further compared with the GMD in  \cite{Bailey12} as that, 
``the GMD is potentially inconsistent if a single component is inconsistent while the AMD is conservative if even a single component is consistent'' [cf. Lemma 1]. Indeed, the union-type AMD fusion is less prone to the problem of misdetection as it does not involve product calculation. More importantly, it does not fuse the information of different targets and of clutter unless they lie to each other too close. 

The AMD of GMs, 
can be easily realized through re-weighting (by using the fusing weights) and combining GMs in neighborhood. Similar idea has actually been applied \cite{Yu16} for pairwise gossip-based fusion and for averaging the ``generalized likelihood'' \cite{Streit08}. Basically, two key issues need to be solved
. First, we need a proper mechanism to design the fusing weights. The most straightforward solution is given by uniform fusing weights, which may not guarantee efficient consensus convergence and appeals primarily to the case when only few P2P communication iterations are allowed. For faster convergence, the popular Metropolis weights \cite{Xiao04, Xiao05} approach is readily competent, given a large number of communication iterations. It determines the fusing weights for the information from sensor $b$ at the host sensor $a$ as follows
\begin{equation}\label{eq:Metropolis} 
\omega_{b\rightarrow a} = \left\{
\begin{array}{ll} 
\frac{1}{1+\max{(|N_a|,|N_b|)}} & \mathrm{if} \hspace{1mm}  b\in N_a, b\neq a \hspace{0.5mm},\\
1-\sum_{l\in N_a}{\omega_{l\rightarrow a}} & \mathrm{if} \hspace{1mm} b=a\hspace{0.5mm},\\
0 & \mathrm{if} \hspace{1mm} b\not\in N_a\hspace{0.5mm},
\end{array} \right.
\end{equation} 
where $N_a$ denotes the set of neighbors of sensor $a$ (excluding $a$).

Second, the AMD of $N$ GCs and $M$ GCs as in \eqref{eq:AMD} will have $N+M$ GCs, which is in general much smaller than $N\times M$ yielded by GCI. Still the local GM size grows linearly with the number of fusing sensors. In order to reduce the number of GCs to be transmitted and to maintain a stable overall GM size, we next present two conservative MR schemes for fusing the gathered T-GMs in a fully distributed fashion. 

\section{Conservative MR schemes}
 \label{sec:protocol}
We use $t\in \mathbb{N}=\{0,1,2,...\}$ to denote the P2P communication iteration. $t=0$ means the original statue of the local sensor without any communication. This section presents two MR approaches in line with the conservative fusion principle, based on either OMR or pairwise GM averaging, which need to be executed at each P2P communication iteration. 

\subsection{Conservative Fusion P.2.1: GM Merging}
The first MR protocol for T-GM fusion is given by combining the newly received and the local T-GMs into one set and merging the close T-GCs based on the proposed OMR. Before this, the GC weights should be scaled by using the fusing weights as addressed, according to their origination sensor. However, as shown in our simulation that this protocol typically bears high communication cost (which increases with $t$) and more iterations ($t>2$) do not yield significantly more benefit, we do not suggest a larger number of communication iterations. Therefore, uniform fusing weights are more preferable (especially for $t\leq2$). 


A key of MR/OMR is to determine the size of gate for fusing GCs to be merged. In our approach, the distance between two T-GCs, e.g., $\mathcal{N}(\mathbf{x};\mathbf{m}_a,\mathbf{P}_a)$ and $\mathcal{N}(\mathbf{x};\mathbf{m}_b,\mathbf{P}_b)$, is measured by the Mahalanobis-type distance 
as follows

\begin{equation} \label{eq:Mahalanobis}
C_{a,b}\triangleq\big(\mathbf{m}_a - \mathbf{m}_b\big) \mathbf{P}^{-1} \big(\mathbf{m}_a - \mathbf{m}_b\big)^\mathrm{T} \hspace{0.5mm},
\end{equation}
where $\mathbf{P}$ is chosen as the covariance of the GC of higher weight. 


A gate threshold $\tau$ is needed to control the GC grouping such that only T-GCs that are of distance smaller than $\tau$ will be merged,
for trade-off 
between the resultant GM size and merging error. 
The gate has a clear physical meaning as it indicates the distance no further than $\tau$ standard deviations from
the state estimate that the real state lies in with a probability, or at least a lower bound on the probability. When the estimate is unbiased and inferred from Gaussian random variables, the probability that the real state $\mathbf{x}$ lies within $\tau$ standard deviations of the state estimate $\mathbf{\hat{x}}$ is given by \cite{Ye00}
\begin{equation} \label{eq:confidence}
\mathrm{Pr}\big[(\mathbf{x}-\mathbf{\hat{x}}) \mathbf{P}^{-1}(\mathbf{x}-\mathbf{\hat{x}}) \big] \leq \gamma \Big(\frac{d}{2},\frac{\tau^2}{2}\Big)\hspace{0.5mm},
\end{equation}
where $\mathbf{P}$ is the error-covariance matrix of the estimate $\mathbf{\hat{x}}$, $\gamma$ is the lower incomplete Gamma distribution and $d$ is
the cardinality of the state vector.

Due to the uniform fusing weights, the T-GM combination and merging will certainly raise the weight sum at sensor $a$ to
\begin{equation} \label{eq:W_rise}
\tilde{W}_a(t)=W_a(t-1)+\sum_{j\in {N_a}}W_{j,\mathrm{T}}(t-1)\hspace{0.5mm},
\end{equation}
where $W_a(t-1)$ and $W_{j,\mathrm{T}}(t-1)$ are the whole GM and the T-GM at local sensor $a$ after $t-1$ iterations of P2P communication, respectively and $W_{j,\mathrm{T}}(0)$ is defined as in \eqref{eq:PHD5}. 

As such, the weights of all GCs $w_a^{(i)}, i=1,2,\cdots, J_a(t)$ after merging at each iteration $t$ need to be re-scaled for correct cardinality estimation, where $J_a(t)$ is the GM size at iteration $t$ and we have $\tilde{W}_a(t) = \sum_{i=1}^{J_a(t)} w_a^{(i)}$. To this end, we may apply average consensus on the cardinality estimates, namely ``cardinality consensus'', which will be carried out simultaneously with the proposed T-GM consensus. This is feasible because the cardinality estimates yielded by the PHD filter \eqref{eq:4} are scalar-valued parameters, for which the standard average consensus based on Metropolis weights \cite{Xiao04, Xiao05} is straightforwardly applicable. 

To this end, the local GM weight sums will also be disseminated in neighborhood along with the T-GCs for consensus and we have the following proposition. 

\textbf{Proposition 2 (Cardinality Consensus)}. The Metropolis weights based average consensus is applied to update the local weight sum at each communication iteration as follows:
\begin{equation}\label{eq:CC}
{W}_a(t) = \sum_{l \in \{ a,N_a\}} \omega_{l\rightarrow a} W_l(t-1)\hspace{0.5mm},
\end{equation}
which will be used for re-scaling the weights of all GCs at each communication iteration $t$, i.e.,
\begin{equation}\label{eq:CC_W_scaling}
w_a^{(i)} \leftarrow \beta_a w_a^{(i)},  \forall i=1,2,\cdots, J_a(t)\hspace{0.5mm}.
\end{equation}
where $\beta_a = \frac{W_a(t)}{\tilde{W}_a(t)}$.

In order to analyze the change of the weight of FA-GCs due to \eqref{eq:CC_W_scaling}, we make two approximate albeit reasonable assumptions: $W_a(t-1) \approx W_b(t-1)$ and $W_{b,\mathrm{T}}(t-1) \approx \alpha W_b(t-1)$, for all $b \in N_a$. Clearly, $\alpha < 1$. As such, \eqref{eq:W_rise} and \eqref{eq:CC} reduce to $\tilde{W}_a(t) \approx (1+ \alpha|N_a|) W_a(t-1)$ and $W_a(t) \approx W_a(t-1)$, respectively. These read
\begin{equation} \label{eq:beta}
\beta_a \approx \frac{1}{1+\alpha|N_a|}
\end{equation}

In most cases, the T-GCs take the majority of the weight sum, namely $\alpha > 0.5$. For example, when sensor $a$ has two neighbors namely $|N_a|=2$, $\beta_a \stackrel{\text{appr}}{<} 0.5$, which indicates that the weight of FA-GCs at local sensor $a$ will be approximately reduced to less than a half by \eqref{eq:CC_W_scaling}. Comparably, the T-GCs merge with many others from neighbors which will counteract such reduction but instead their weight will likely be increased slightly. That is, the target-likely signal will be enhanced while the FA-suspicious signal will be weakened or even ultimately removed by pruning. This will give rise to the SNR at local sensors, reducing the possibility of causing false alarms and facilitating more accurate estimation.
We refer to this fusion protocol as \textit{conservative GM merging} (CGMM). 

For illustrative purposes, CGMM operations including GC selection, transmission, merging and re-weighting are given, as shown in Fig.1, for a 1-dimensional state space model using two sensors. In the top row, the original PHDs at local sensors are given as GMs, each having two significant GCs that likely correspond to targets. The sensors share them with each other, and then GC merging (and pruning to remove very insignificant FA-GCs) and re-weighting are performed, as shown in the middle and bottom rows, respectively. The resulting GMs are reweighted such that they have the same weight sum for cardinality consensus. At the end, the T-GCs will become more significant due to merging while the FA-GCs are weakened, leading to an enhanced SNR. 

\begin{figure}\label{fig:1}
\centering
\includegraphics[width=8 cm]{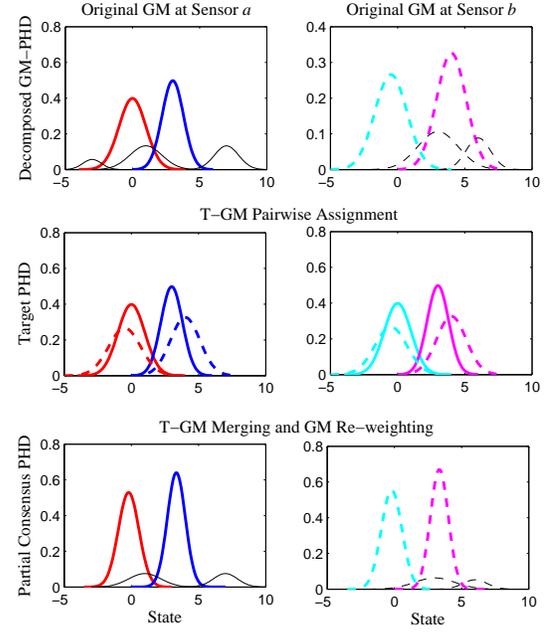}
\caption{Illustration of the proposed CGMM fusion between two sensors. The GCs originally formed by sensor $a$ are given in solid lines, while those formed by sensor $b$ are given in dash lines. Significant components of the GM are given in color while the insignificant ones in black}
\end{figure}

\subsection {Conservative Fusion P.2.2: Pairwise GM Averaging}

CGMM can not guarantee all received GCs to be merged to the local T-GM unless a sufficiently high merging threshold $\tau$ is used which will in turn cause greater merging errors. As a result, the local T-GM size will likely grow against the P2P communication iteration
. As an alternative, we integrate the received T-GM to the local GM in a way such that each of the received T-GC is fused to the nearest host GC immediately if closely enough or otherwise abandoned. This will retain a promisingly constant local GM size during networking. To this end, we associate the received T-GMs from neighbors with the host T-GM based on Hungarian assignment (also called Munkres algorithm \cite{Munkres57}) and gating. Then, only associated GCs will be fused in the manner of ``averaging''. 

For clarity, denote the host sensor as $a$, one of its neighbors as $b \in N_a$, and the number of original T-GCs as $n_a$ and $n_b$, respectively. To carry out Hungarian assignment, a $n_a \times n_b$ cost matrix needs to be constructed as follows: if $n_a \leq n_b$ (otherwise transpose the matrix)
\begin{equation} \label{eq:26}
\left[ \begin{array}{ccc}
C_{1,1} & \cdots & C_{1,n_b} \\
\cdots & \cdots & \cdots \\
C_{n_a,1} & \cdots & C_{n_a,n_b} \\
\end{array} \right] \hspace{0.5mm},
\end{equation}
where 
$C_{i,j}$ is the Mahalanobis-type distance as in \eqref{eq:Mahalanobis} between GC $i$ from sensor $a$ and GC $j$ from sensor $b$.


The optimal assignment is given by choosing one entry at each row of the cost matrix \eqref{eq:26}, all entries belonging to different columns, with a minimal sum. That is, the optimization cost function is given by 

\begin{equation} \label{eq:28}
\underset{\pi \in \Pi_{n_b}}{\mathrm{argmin}} \sum_{i=1}^{n_a} {C_{i,\pi(i)}} \hspace{0.5mm},
\end{equation}
where $\Pi_{n_b}$ is permutations of $n_b$ entries and $\pi(i)$ indicates the $i$th entry in the permutation $\pi$.

The Hungarian algorithm has proven to be efficient in solving the above assignment problem in polynomial time \cite{Munkres57}. As a result, all the GCs in the smaller GM set will be assigned to one and only one GC from the larger GM set, while the GCs from the latter will be assigned to one or no GC from the former. We call this \textit{one to one-or-zero} assignment where the unassigned component will be unfused. 

Furthermore, a double-checking step is required so that only the assigned pair that are close enough are to be fused. Again, we use the Mahalanobis-type distance as given in \eqref{eq:Mahalanobis} to measure the distance between two GCs and the rule \eqref{eq:confidence} to design the gating threshold. 
Any assignment that does not fall in the valid gate will be canceled. 
The unassigned T-GCs will be abandoned and will not be involved in any fusion if it is received from the neighbor, otherwise it remains unchanged at local sensor. This will guarantee promisingly that the GM set of constant size at each local sensor. Finally, the pairwise assigned GCs from 
will be fused in the manner of Metropolis weights-based averaging at each iteration, 
as given below. 

First, Metropolis weights are used to re-scale all T-GC weights according to their origination sensor. Then, the associated GCs are labeled with the same, to say $\ell$, originating from whether sensor $a$ or its neighbors $N_a$. As addressed above, each sensor contributes a maximum of one GC to each group. 
We denote all the sensors that contribute one GC to group $\ell$ by a set $S_a^{[\ell]} \in \{a, N_a\}$ for which we have the following proposition for conservative fusion. 

\textbf{Proposition 3 (GM averaging)} All T-GCs associated in $S_a^{[\ell]}$ are averaged, resulting in a new single GC $\mathcal{N}(\mathbf{x};\mathbf{m} _a^{[\ell]}(t),\mathbf{P}_a^{[\ell]}(t))$ with weight $w_a^{[\ell]}(t)$ as follows:
\begin{equation}\label{eq:Aver_w}
w_a^{[\ell]}(t)= \frac{\sum_{l \in S_a^{[\ell]}} \omega_{l\rightarrow a} w_l^{[\ell]}(t-1)}{\sum_{l \in S_a^{[\ell]}} \omega_{l\rightarrow a}}\hspace{0.5mm},
\end{equation}
\begin{equation}\label{eq:Aver_x}
\mathbf{m}_a^{[\ell]}(t) = \frac {\sum_{l \in S_a^{[\ell]}} \omega_{l\rightarrow a} w_l^{[\ell]}(t-1)\mathbf{m}_l^{[\ell]}(t-1)} {\sum_{l \in S_a^{[\ell]}} \omega_{l\rightarrow a} w_l^{[\ell]}(t-1)}\hspace{0.5mm},
\end{equation}
\begin{equation}\label{eq:Aver_P}
\mathbf{P}_a^{[\ell]}(t) = \mathbf{P}_\mathrm{OMR} 
\hspace{0.5mm},
\end{equation}
where $\mathbf{P}_\mathrm{OMR}$ is given as in \eqref{eq:minGM_P} by substituting $\mathcal{I} = S_a^{[\ell]}$ and $\tilde{\mathbf{P}}_i= \mathbf{P}_i^{[\ell]}(t-1) + \big(\mathbf{m}_i^{[\ell]}(t-1)-\mathbf{m}_a^{[\ell]}(t)\big) \big(\mathbf{m}_i^{[\ell]}(t-1)-\mathbf{m}_a^{[\ell]}(t)\big)^\mathrm{T}$ for all $i \in S_a^{[\ell]}$. 



As shown, the calculation of the fused state and covariance is akin to that of CGMM but they are different at the fused weight, which is an average in \eqref{eq:Aver_w} rather than a sum in \eqref{eq:GMmerging_w} in the CGMM. Therefore, \eqref{eq:W_rise} does not hold here but instead roughly $\tilde{W}_a(t) \approx W_a(t-1)$ and $\beta_a \approx 1$ instead of \eqref{eq:beta}. Comparably speaking, the FA-GCs will not be so significantly weakened in CGMA as in CGMM. However, the cardinality average consensus scheme as given in Proposition 2 can still be applied at each communication iteration to re-weight all GCs including the averaged one given in \eqref{eq:Aver_w} and the insignificant GC that is not involved in fusion
. Overall, we refer to this consensus protocol as \textit{conservative GM averaging} (CGMA).

\subsection{Potential Extensions}
The proposed union-type, conservative fusion and partial-consensus-based distributed GM-PHD fusion can be extended in terms of both the communication protocol and the local filter. In the former, other consensus protocols other than averaging schemes (e.g., diffusion \cite{sayed14,Dedecius17}, flooding \cite{Li17flooding}) can be applied, while in the latter, multi-Bernoulli filters \cite{Guldogan14, Wang17, Suqi17} and even particle filter-based RFS filters can be employed based on novel mixture reduction or particle-resampling schemes. 

\section{Simulations}
\label{sec:simulation}
In this section, the proposed CGMM and CGMA approaches for distributed GM-PHD fusion are evaluated for tracking either a single target or simultaneous multiple targets, with comparison to the benchmark GCI/KLA fusion \cite{Battistelli13} and the pure cardinality consensus based on either flooding (CCF) \cite{Li17flooding} or Metropolis weights-based averaging (CCA). 
These different distributed filters will be evaluated on the same ground truths, sensor data series and sensor network setting up.

For MR in all filters: GCs with a weight lower than $10^{-4}$ will be truncated, any two GCs closer than Mahalanobis-distance $\tau =5$ will be merged, and the maximum number of GCs is 50 in the case for tracking a single target and 100 in the case of multiple targets. The proposed partial consensus is carried out based on the rank rule P.1.1 for selecting the T-GCs. To save communication in GCI, we suggest a threshold $w_c = 0.005$ 
such that only the GC with a weight larger than $w_c$ will be disseminated to neighbors and then be considered in the subsequent fusion. 

The optimal sub-pattern assignment (OSPA) metric \cite{Schuhmacher08} is used to evaluate the estimation accuracy of the filter, with cut-off parameter $c=1000$ and order parameter $p=2$; for the meaning of these two parameters, please refer to \cite{Schuhmacher08}. We refer to the average of OSPAs obtained by all sensors in the network at each sampling step as \textit{Network OSPA}. The average of the Network OSPAs over all filtering steps is called \textit{Time-average Network OSPA}. 
To evaluate the communication cost, we record a GC that consists of a weight parameter (1 tuple), a 4-dimensional vector mean (4 tuples), and a 4$\times$4 -dimension matrix covariance (16 tuples) as data size 21 tuples and the scale-valued cardinality parameter as 1 tuple. Given that the covariance matrix is symmetric, only 10 tuples are needed here and then a GC only needs 15 tuples for data storing. Furthermore, to measure the efficiency of different consensus protocols, we define a consensus efficiency (CE) measure regarding the average OSPA reduction gained by sharing each tuple of network data as follows:
\begin{equation} \label{eq:CE}
\mathrm{CE}\triangleq\frac{\text{OSPA reduction due to communication}}{\text{Network communication cost (no. tuples)}}\hspace{0.5mm}.
\end{equation}

\begin{figure}\label{fig:2}
\centering
\includegraphics[width=7 cm]{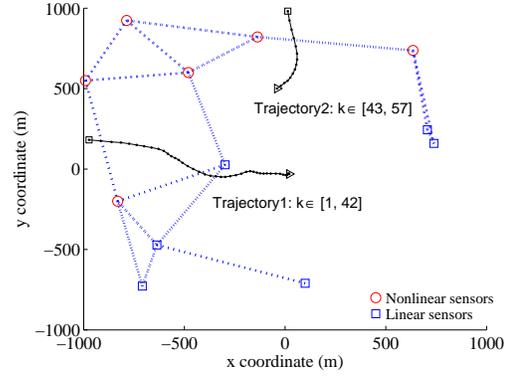}
\caption{Tracking scenario, target trajectories and sensor network}
\end{figure}

The simulations are set up in a scenario over the planar region $[-1000,1000]\mathrm{m}\times [-1000,1000]\mathrm{m}$ which is monitored by a randomly generated sensor network (with total 12 sensors and diameter $6$) as shown in Fig.2. We assume two different ground truths for the target trajectories, to be presented in the following two subsections respectively. To capture the average performance, we perform each simulation 100 MC runs with independently generated observation series for each run. Different numbers $t$ of P2P communication iterations from 0 (without applying any information disseminating) to 12 (twice the network diameter) are applied to all consensus schemes. To set up the local filter, the ground truth is simulated as follows: The target birth process follows a Poisson RFS with intensity function $ \gamma_k(\mathbf{x})= \sum_{i=1}^4 \lambda_i\mathcal{N}(.;\mathbf{m}_i,\mathbf{Q}_r)$, where the Poisson rate parameters $\lambda_1=\lambda_2=\lambda_3=\lambda_4$, the Gaussian parameters $\mathbf{m}_1=[0,0,0,0]^\mathrm{T}$, $\mathbf{m}_2=[-500, 0, -500, 0]^\mathrm{T}$, $\mathbf{m}_3=[0,0,500,0]^\mathrm{T}$, $\mathbf{m}_4=[500, 0, -500, 0]^\mathrm{T}$, $\mathbf{Q}_r=$ diag$([400,100,400,100]^\mathrm{T})$, and diag$(\mathbf{a})$ represents a diagonal matrix with diagonal $\mathbf{a}$. In addition, the target intensity function spawn from target $\mathbf{u}$ is given as $ b_k(\mathbf{x} | \mathbf{u})= 0.05\mathcal{N}(.;\mathbf{u},\mathbf{Q}_b)$, where $\mathbf{Q}_b=$ diag$([100,400,100,400]^\mathrm{T})$. Each target has a time-constant survival probability $p_S(\mathbf{x}_k)=0.99$ and the survival target follows a nearly constant velocity motion as given
\begin{equation}\label{eq:GMPHD5}
\mathbf{x}_k= \left[ \begin{array}{cccc}
1 & \Delta & 0 & 0 \\
0 & 1 & 0 & 0 \\
0 & 0 & 1 & \Delta \\
0 & 0 & 0 & 1 \\
\end{array} \right] \mathbf{x}_{k-1}+ \left[ \begin{array}{cc}
\Delta^2/2 & 0 \\
\Delta & 0 \\
0 & \Delta^2/2 \\
0 & \Delta \\
\end{array} \right] \mathbf{u}_k\hspace{0.5mm},
\end{equation}
where $\mathbf{x}_k=[p_{x,k},\dot{p}_{x,k},p_{y,k},\dot{p}_{y,k}]^\mathrm{T}$ with the position $[p_{x,k},p_{y,k}]^\mathrm{T}$ and the velocity $[\dot{p}_{x,k},\dot{p}_{y,k}]^\mathrm{T}$, the sampling interval $\Delta =1$s, and the process noise $\mathbf{u}_k \sim \mathcal{N}(\mathbf{0}_2,25\mathbf{I}_2)$.

Without loss of generality, we employ a hybrid sensor network that consists of both linear and nonlinear observation sensors which run linear GM-PHD filter and unscented transform based nonlinear GM-PHD filter \cite{Vo06}, respectively. The sensors are ordered from 1 to 12, where the sensors no.1-6 generate linear observation (which are referred to as linear sensors, marked by square in Fig.2) while the rest (no. 7-12) generate nonlinear observation (referred to as nonlinear sensors, marked by circles in Fig.2). The linear sensors have the same time-constant target detect probability $p_D(\mathbf{x}_k)=0.95$ and the linear position observation model given as follows
\begin{equation}\label{eq:GMPHD6}
\mathbf{z}_k= \left[ \begin{array}{cccc}
1 & 0 & 0 & 0 \\
0 & 0 & 1 & 0 \\
\end{array} \right] \mathbf{x}_k+ \left[ \begin{array}{c}
v_{k,1} \\
v_{k,2} \\
\end{array} \right]\hspace{0.5mm},
\end{equation}
with $v_{k,1}$ and $v_{k,2}$ as mutually independent zero-mean Gaussian noise with the same standard deviation of 10. 

The FOV of each nonlinear sensor is a disc of radius 3000m centralized with the sensor's position $[s_{n,x},s_{n,y}]^\mathrm{T},n=16,..,30$, which is able to fully cover the scenario. The target detection probability depends on the target position $[p_{x,k},p_{y,k}]^\mathrm{T}$, as given by $p_D(\mathbf{x}_k)=0.95 \mathcal{N}([|p_{x,k}-s_{n,x} |,|p_{y,k}-s_{n,y}|]^\mathrm{T}; \mathbf{0},6000^2 I_2)/\mathcal{N}(0;0,6000^2I_2)$ and the nonlinear range and bearing observation is given by
\begin{equation} \label{eq:GMPHD7}
\mathbf{z}_k= \left[ \begin{array}{c}
\sqrt{(p_{x,k}-s_{n,x} )^2+(p_{y,k}-s_{n,y})^2}\\
\arctan\big((p_{y,k}-s_{n,y})/(p_{x,k}-s_{n,x})\big) \\
\end{array} \right] + \mathbf{v}_k\hspace{0.5mm},
\end{equation}
where $\mathbf{v}_k \sim \mathcal{N}( ;\mathbf{0},\mathbf{R}_k )$, with $\mathbf{R}_k=$ diag$\left([\sigma_r^2,\sigma_\theta^2 ]^\mathrm{T}\right), \sigma_r=10$m, $\sigma_\theta=\pi/90$ rad/s.

Clutter is uniformly distributed over each sensor's FOV with an average rate of $r$ points per scan. For the nonlinear sensors, we set $r=5$ in both scenarios indicating a clutter intensity $\kappa_k = 5/3000/2\pi$ while for th linear sensors we set $r=5$ for the first single target scenario and $r=10$ for the second multiple target scenario, indicating clutter intensities $\kappa_k = 5/2000^2$ and $\kappa_k = 10/2000^2$, respectively. 



\begin{figure*}\label{fig:3}
\centering
\includegraphics[width=16 cm]{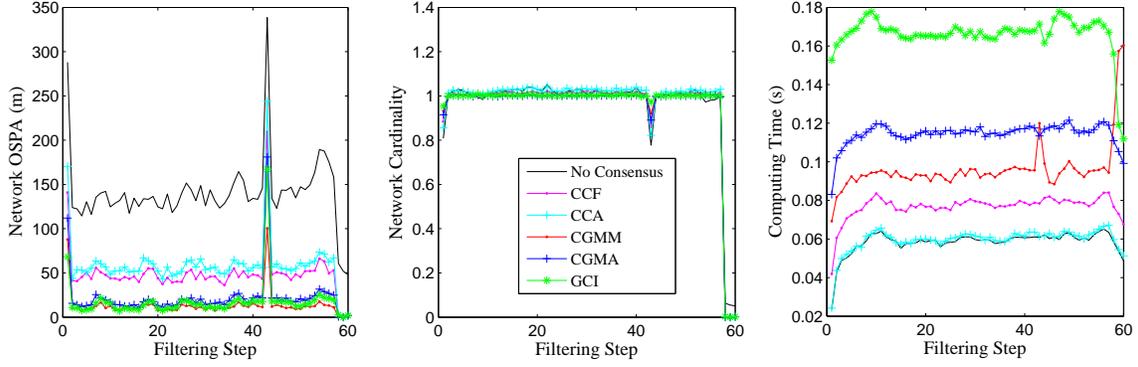}
\caption{Network OSPA, online estimated number of targets and computing time of different consensus protocols for each filtering step when six iterations of P2P communication are applied}
\end{figure*}

\begin{figure*}\label{fig:4}
\centering
\includegraphics[width=16 cm]{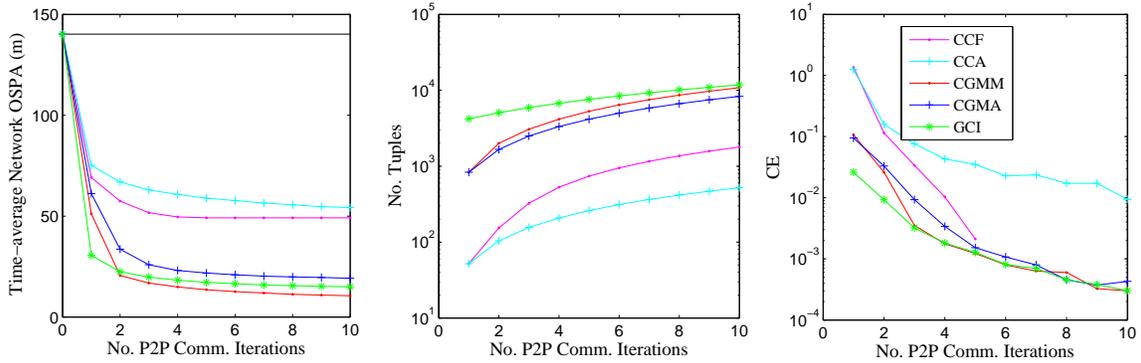}
\caption{Time-averaged network OSPA, network communication cost and CE against P2P communication iterations}
\end{figure*}

\subsection{Single Target Scenario}

First, we limit the maximal number of targets that simultaneously exist in the scenario to one for generating the ground truth as that new target which can only appear after the existing target disappears. Also, there is no target spawning. That is to say, the tracking at any time actually involves maximally one target, which is favorable for CI/GCI. The network and the ground truth of the target trajectories are given in Fig.2. 

When a total of $t=6$ P2P communication iterations are applied, the Network OSPA, the online estimated number of targets, and the computing time of different consensus protocols for each filtering step are given in Fig.3, separately. For different numbers of communication iterations, the time-averaged network OSPA, time-averaged network communication cost, and CE of different consensus protocols are given in Fig.4, separately. We have the following key findings
:

\begin{enumerate}
\item All consensus schemes converge with the increase of the number of P2P communication iterations; meanwhile, the more iterations, the higher the communication and computing cost and the lower OSPA. In particular, when $t=1$ (each sensors only share information with their immediate neighbors), GCI yields the best performance, providing the lowest Network OSPA and time-average OSPA over all. When $t\geq 2$ , CGMM yields the lowest OSPA over all which is even better than the GCI.
\item When $t=6$, the computing time required by GCI is the most and is much higher than that of the others, while CGMA and CGMM come as the second and the third, respectively.
\item On communication cost, CGMM costs slightly more than CGMA, both smaller than GCI, especially when few P2P communication iterations ($t<6$) are applied. 
\item Cardinality consensus has improved the cardinality estimation in all consensus schemes. However, the benefit of pure cardinality consensus is limited, whether CCF by flooding or CCA by averaging, which converges to a level that is significantly inferior to that of the others including GCI, CGMM, and CGMA. However, this is achieved at the price of significantly less computation and communication. 
\item The CE decreases with the increase of $t$ in all consensus schemes. 
Overall, CCA yields the highest CE and CCF comes second; comparably, CCF converges faster at the expense of more communication cost than CCA. When $t\geq 6$, the CCF achieves complete consensus/convergence \cite{Li17flooding} and so it will no further reduce the OSPA, leading to a zero CE. In regard to CE, CGMA slightly outperforms GCI and CGMM while the latter two perform similar. 
\item Local GM size remains constant favorably during network communication at each filtering step using CCF, CCA, and CGMA but varies (mainly increases with $t$) due to CGMM and GCI. This is one advantage that CGMA has over CGMM and GCI.
\end{enumerate}

To summarize the results in the single target case, CGMM is a fair alternative to GCI in favor of smaller OSPA and fewer fusion computation and communication, while CGMA is a better choice than GCI in favor of less computation and communication, and higher CE. 
More discussion will be given in Section \ref{sec:simulation}.C.

\subsection{Multiple Target Scenario}
In this case, we extend the maximal number of targets that simultaneously exist in the scenario to three for generating a new ground truth. The trajectories of totally four targets are given in Fig.5 with the starting and ending times of each trajectory noted. 

To show the simulation result, similar contents given in Figs. 6 and 7 correspond to those in Figs. 3 and 4, respectively. While some of them give similar indication, e.g., the relative communication and computation cost of different protocols, key new findings are summarized as follows.

\begin{figure}\label{fig:5}
\centering
\includegraphics[width=7.5 cm]{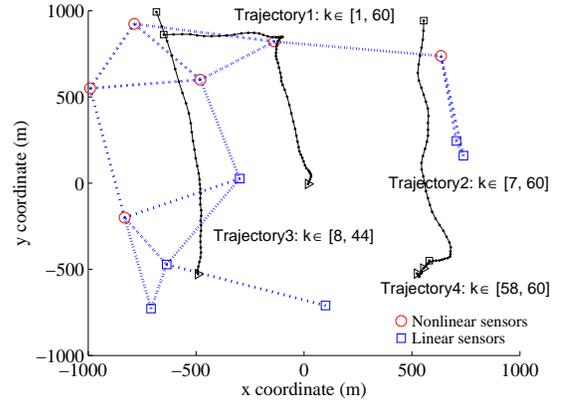}
\caption{Trajectories of simultaneously appearing multiple targets}
\end{figure}

\begin{figure*}\label{fig:6}
\centering
\includegraphics[width=16 cm]{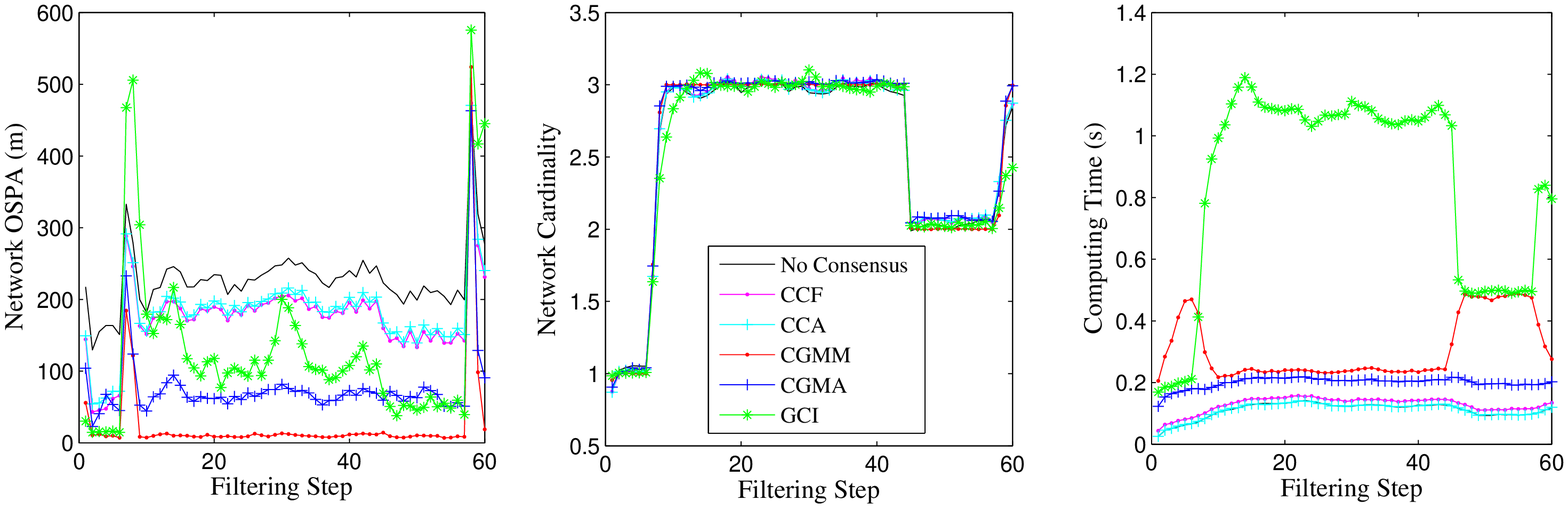}
\caption{Network OSPA, online estimated number of targets and computing time of different consensus protocols for each filtering step when six iterations of P2P communication are applied}
\end{figure*}

\begin{figure*}\label{fig:7}
\centering
\includegraphics[width=16 cm]{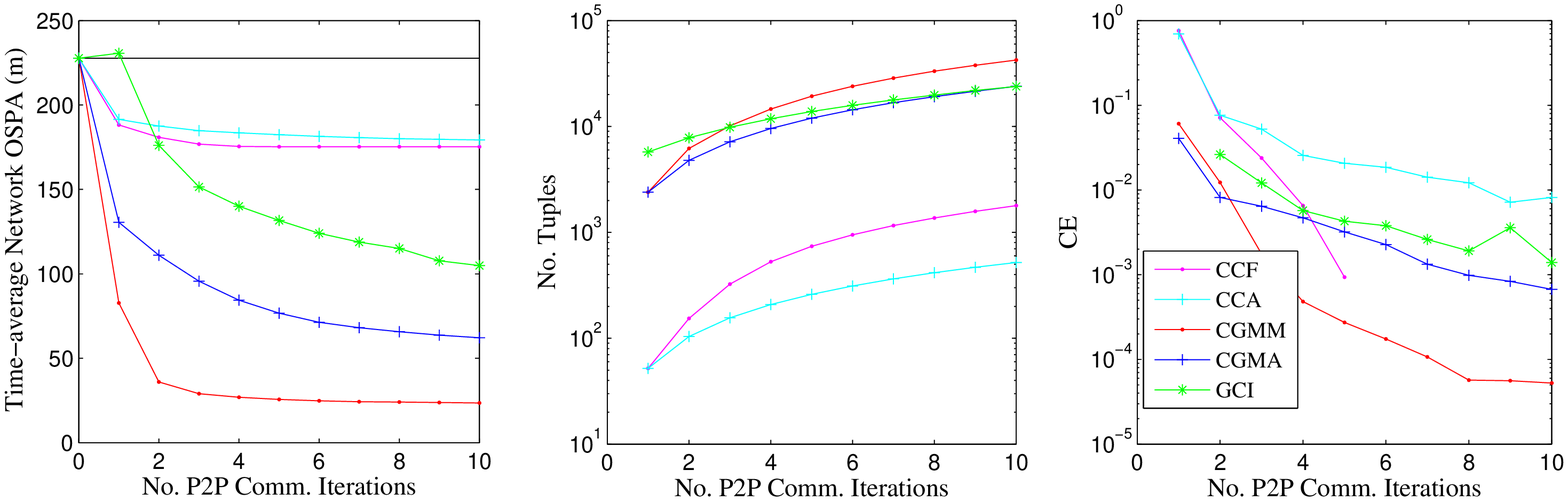}
\caption{Time-averaged network OSPA, network communication cost and CE for different numbers of consensus iterations}
\end{figure*}

\begin{enumerate}
\item On filtering accuracy, CGMM gets the minimum network OSPA which significantly outperforms the others and CGMA comes second. In particular when $t=1$, CGMM that consumes the same communication as CGMA and smaller than GCI, yields the largest OSPA reduction, even more significant than that of the others by performing multiple iterations of communication. This simply indicates that, only immediate neighborhood T-GM information sharing for partial consensus outperforms sophisticated averaging fusion like GCI based on multiple iterations for complete consensus. We refer to this as ``\textit{many could be better than all}''. 
\item Somehow surprisingly, GCI does not benefit the filter accuracy when $t=1$. We leave here an attempt to explain the reason.
\item On communication, 
CGMM costs less than GCI if $t\leq 3$ but more if otherwise. CGMA communicates always less than CGMM and GCI.
\item On computation, GCI costs the most again while CGMM and CGMA perform similar except when the true number of targets is one for which CGMM may be more computing costly than GCI.
\item GCI shows delay at detecting new born targets as it will even increase OSPA compared with the centralized filter with no sensor communication at time $k\in [7,11]$ and $k\in [58,60]$ as shown in the middle sub-figure of Fig.6. 
This obtuse capability in new target detection is indeed unfavorable, significantly reducing the filtering accuracy (as shown in the left sub-figure of Fig.6). 
\item CGMA and CGMM perform worse on CE than GCI except the first communication iteration, all inferior to CCF and CCA again. But, their achievements in reducing OSPA is to a large degree more significant than that of GCI and others. 
\end{enumerate}

To summarize this multiple-target case, both CGMA and CGMM (in particular) afford better alternative to GCI in favor of smaller OSPA, less fusion computation and even less communication for the same OSPA reduction gain.

\subsection{Further Discussion}
Experimental findings reported in the literature are notable. The performance of GCI is greatest for few sensors and distant targets \cite{Uney13, Battistelli14a} or only a single target \cite{Battistelli13, Battistelli14b, Guldogan14}. Closely-distributed targets in dense clutter environment have not been particularly considered except few works such as \cite{Battistelli15, Wang17}, which just showed that GCI made worse result when targets are close. For example, the cardinality estimation is worse at around time $k=800$s when more iterations of GCI fusion are applied, as shown in Fig.5-7 of \cite{Battistelli15}. Delay has also been observed in estimating the number of targets when new targets appear in the scenario in \cite{Gunay16}. More specifically, the simulation given in Section V.A of \cite{Wang17} has explicitly demonstrated that GCI will degrade the local PHD filter in the case of close targets whose distance is under a specific threshold and/or in the case of low SNR. Deficiency of GCI for handling misdetection has been particularly noticed in \cite{Yu16,Yi17}. Relatively, the findings given in \cite{Hwang04} suggested that the arithmetic average method is most robust to incorrect information than the geometric average. It has also been demonstrated that the CI provides estimation error covariance that is not honest but pessimistic for track fusion with feedback, inferior to the minimum variance rule \cite{Mori12}. 

In summary of our findings and those given in the literature, the problems that a distributed multi-target filter may potentially suffer from due to GCI include: 
 
\begin{enumerate}
\item Weakness to deal with closely distributed targets and/or low SNR background;
\item Prone to mis-detection or local sensor failure;
\item Delay in detecting new appearing targets;
\item High communication and computation cost for complete consensus.
\end{enumerate}
	
It seems still unclear how to fix these problems on the basis of GCI, even some of the causes have been noted, nor was that our intension in this paper. We leave here direct simulation demonstration about the failures of GCI in complicated multi-target scenarios (e.g., new targets appear frequently, targets move closely or there is a high rate of mis-detection or clutter) in which our proposed approaches demonstrate more significant advantage. In particular, the straightforward arithmetic average based CGMM that can be easily implemented on different filter beds yields significant accuracy benefit with only one or two iterations of P2P diffusion. 

The merit of the presented partial consensus and conservative arithmetic average fusion is not only on reliable and significant consensus benefit, but also on inexpensive communication and computation for complying with the need of real time filtering. It is very crucial to note that, a key challenge in many large-scale WSN scenarios comes exactly from limitations imposed on the communication bandwidth/power allowance and the sensor computing capability because the nodes are low powered wireless devices.

\section{Conclusion} 
\label{sec:conclusion}
For distributed GM-PHD fusion, this paper has proposed a notion of``partial consensus'' which abandons the ultimate goal that the estimate of each sensor converges to the estimation conditioned on all the information over the entire network but instead neighboring sensors shares only highly-weighted GCs with each other and at the end, the network achieves partially consensus. In addition to saving communication and computation, the local SNR at each sensor can be increased because of partial consensus, reducing the possibility to generate false alarms and facilitating more accurate estimation. To further reduce the communication cost, the disseminated significant GCs can be either pairwise averaged or locally merged in a fully distributed and conservative manner. In parallel, the arithmetic average consensus is sought on the GM weight sum at each communication iteration. 

Simulations based on both single target scenario and multiple target scenario have been provided to demonstrate the effectiveness and reliability of our approach with comparison to the GCI, which is the state of the art approach for distributed RFS filter fusion. Although the GCI works well in the single target scenario in the presence of low misdetection and clutter rates, it exhibits severe problems in complicated multi-target scenarios, such as delay in detecting new appearing targets, and incompetent to handle closely distributed targets, intensive clutter and mis-detection, in addition to its high communication and computation cost. 

For multi-target density fusion in the presence of significant clutter and misdetection, our final remarks are:
\begin{itemize}
\item \textit{Many could be better than all}: the concept of ``partial consensus'' is important as can not only save communication and computation but also benefit the accuracy more than the complete consensus. %
\item \textit{Union outperforms intersection}: Union-format arithmetic average fusion, as the original average consensus is, is computationally easier and provably more reliable than the Intersection-format geometric average fusion, while the former is also more conservative in general.
\end{itemize}

\ifCLASSOPTIONcaptionsoff
 \newpage
\fi



\bibliographystyle{IEEEtran}
\bibliography{PartialConsensus}
\end{document}